\documentclass{article}

\usepackage[a4paper,left=2.5cm, right=2.5cm, top=2.5cm, bottom=2.5cm]{geometry}

\usepackage{microtype}
\usepackage{graphicx}
\usepackage{subfigure}
\usepackage{float}

\usepackage{hyperref}

\usepackage{amsmath}
\usepackage{natbib}
\usepackage{amssymb}
\usepackage{amsfonts}
\usepackage{latexsym}
\usepackage{amsthm}

\newcommand{\edp}{{\mathcal{P}_{\e,\delta}}}  
\newcommand{\tlp}{{\mathcal{P}_{\text{TLap}}}}  
\newcommand{\ft}{{f_{\text{TLap}}}}  

\newcommand{\opta}{{V_1^*}}    
\newcommand{\optb}{{V_2^*}}    

\newcommand{\lowa}{{V_1^{low}}}    
\newcommand{\lowb}{{V_2^{low}}}    

\newcommand{\uppa}{{V_1^{upp}}}    
\newcommand{\uppb}{{V_2^{upp}}}    

\newcommand{\database}{{\mathcal{D}}}    
\newcommand{\KM}{{\mathcal{K}}}    
\newcommand{\e}{{\epsilon}}    
\newcommand{\loss}{{\mathcal{L}}}    
\newcommand{\D}{{\Delta}}    
\newcommand{\dD}{{\tilde{\Delta}}}    
\newcommand{\p}{{\mathcal{P}}}  
\newcommand{\dpp}{{\tilde{\mathcal{P}}_N}}  
\newcommand{\dll}{{\tilde{\mathcal{L}}_N}}  

\newcommand{\R}{{\mathbb{R}}}  
\newcommand{\Z}{{\mathbb{Z}}}  

\newtheorem{theorem}{Theorem}
\newtheorem{lemma}{Lemma}
\newtheorem{definition}{Definition}

\begin{document}
\title{ Privacy and Utility Tradeoff in Approximate Differential Privacy}
\date{}
\author{
  Quan Geng, Wei Ding, Ruiqi Guo, and Sanjiv Kumar\\
  \\
 Google AI \\
New York, NY 10011 \\
Email: {qgeng, vvei, guorq, sanjivk}@google.com
}
\maketitle
\begin{abstract}
We characterize the minimum noise amplitude and power for noise-adding mechanisms in $(\epsilon, \delta)$-differential privacy for single real-valued query function. We derive new lower bounds using the duality of linear programming, and new upper bounds by proposing a new class of $(\epsilon,\delta)$-differentially private mechanisms, the \emph{truncated Laplacian} mechanisms. We show that the multiplicative gap of the lower bounds and upper bounds goes to zero in various high privacy regimes, proving the tightness of the lower and upper bounds and thus establishing the optimality of the truncated Laplacian mechanism. In particular, our results close the previous constant multiplicative gap in the discrete setting. Numeric experiments show the improvement of the truncated Laplacian mechanism over the optimal Gaussian mechanism in all privacy regimes.
\end{abstract}

\section{Introduction} \label{sec:intro}
Differential privacy, introduced by \citet{DMNS06}, is a framework to quantify to what extent individual privacy in a statistical dataset is preserved while releasing useful aggregate information about the dataset.
Differential privacy provides strong privacy guarantees by requiring the near-indistinguishability of whether an individual is in the dataset or not based on the released information. 
For more motivation and background of differential privacy, we refer the readers to the survey by \citet{DPsurvey} and the book by \citet{DPbook}.

Since its introduction, differential privacy has spawned a large body of research in differentially private data-releasing mechanism design, and the noise-adding mechanism has been applied in many machine learning algorithms to preserve differential privacy, e.g., 
logistic regression \citep{CM08}, 
empirical risk minimization \citep{ERM,ERM2}, 
online learning \citep{Jain12}, 
statistical risk minimization \citep{Duchi12},
statistical learning \citep{PAC}, 
deep learning \citep{Shokri15, Abadi2016, Phan2016}
distributed optimization \citep{Agarwal18}, 
hypothesis testing \citep{HT18},
matrix completion \citep{JainMC},
expectation maximization \citep{EM},
and principal component analysis \citep{PCA, PCA2}.

The classic differential privacy is called $\e$-differential privacy, which imposes an upper bound $e^\e$ on the multiplicative distance of the probability distributions of the randomized query outputs for any two neighboring datasets. The standard approach for preserving $\e$-differential privacy is adding a noise with the Laplacian distribution to the query output.
Introduced by \citet{DKMMN06}, the approximate differential privacy is $(\e,\delta)$-differential privacy, and the common interpretation of $(\e,\delta)$-differential privacy is that it is $\e$-differential privacy ``except with probability $\delta$'' \citep{RenyiDP}.
The standard approach for preserving $(\e,\delta)$-differential privacy is the Gaussian mechanism, which adds a Gaussian noise to the query output.

To fully make use of the differentially private mechanisms, it is important to understand the fundamental trade-off between privacy and utility (accuracy). For example, within the class of noise-adding mechanisms, given the privacy constraint $\e$ and $\delta$, we are interested in deriving the minimum amount of noise added to achieve the highest accuracy and utility while preserving the differential privacy. In the literature, there have been many works on optimal differential privacy mechanism design and characterizing the privacy and utility tradeoff in differential privacy. For a single count query function under $\epsilon$-differential privacy, \citet{Ghosh09} show that the geometric mechanism is universally optimal under a Bayesian framework, and \citet{minimax10} derived the optimal noise probability distributions under a minimax cost framework.
\citet{GV_IT_Epsilon} show that the optimal noise distribution has a staircase-shaped probability density function for single real-valued query function under $\e$-differential privacy, and \citet{GV_2_Dimension} generalized the result to two-dimensional query functions. \citet{DomingoFerrer2013} also independently derived the staircase-shaped noise probability distribution under a different optimization framework.

\citet{GV_IT_Approximate} show that for a single integer-valued query function under $(\epsilon, \delta)$-differential privacy, the discrete uniform noise distribution and the discrete Laplacian noise distribution are asymptotically optimal within a constant multiplicative gap in the high privacy regions.
\citet{icmlGaussian} improved the classic analysis of the Gaussian mechanism for $(\e,\delta)$-differential in the high privacy regime ($\e \to 0$), and developed an optimal Gaussian mechanism whose variance is calibrated directly using the Gaussian cumulative density function instead of a tail bound approximation.

\subsection{Our Contributions}

In this work, we characterize the minimum noise amplitude and power for noise-adding mechanisms in $(\e, \delta)$-differential privacy for single real-valued query function. Our contributions are three-fold:

First, we present a new class of $(\e, \delta)$-differentially private noise-adding mechanisms, \emph{truncated Laplacian} mechanisms. Applying the truncated Laplacian mechanism, we derive new achievable upper bounds on minimum noise amplitude and noise power in $(\e,\delta)$-differential privacy for single real-valued query function. The key insights from the new mechanisms design are that the noise probability density function shall decay as fast as possible while being $\e$-differentially private when the noise is small, and then sharply reduce to zero when the noise is big, to avoid a heavy tail distribution which would incur a high cost.

Second, we derive new lower bounds on the minimum noise amplitude and minimum noise power. The key technique is to discretize the continuous probability distribution and the loss function, and transform the continuous functional optimization problem to linear programming. Applying the lower bound result in \citet{GV_IT_Approximate} for \emph{integer-valued} query function, which is based on the duality of linear programming, we derive new lower bounds for \emph{real-valued} query functions under $(\e, \delta)$-differential privacy.

Third, we show that the multiplicative gap of the lower bounds and upper bounds goes to zero in various high privacy regimes, proving the tightness of the lower and upper bounds, and thus establish the optimality of the truncated Laplacian mechanism for minimizing the noise amplitude and noise power under $(\e, \delta)$-differential privacy. In particular, our result closes the previous constant multiplicative gap between the lower bound and the upper bound (using discrete uniform distribution and discrete Laplacian distribution) in \citet{GV_IT_Approximate}.

Comprehensive numeric experiments show the improvement of the truncated Laplacian mechanism over the optimal Gaussian mechanism in \citet{icmlGaussian} by significantly reducing the noise amplitude and noise power in all privacy regimes.

\subsection{Organization}

The paper is organized as follows.
In Section~\ref{sec:formulation}, we give some preliminaries on differential privacy, and derive the $(\e,\delta)$-differential privacy constraint on the additive noise probability distribution and define the minimum noise amplitude and noise power under $(\e,\delta)$-differential privacy. 
Section~\ref{sec:upper_bound} presents the truncated Laplacian mechanism for preserving $(\e, \delta)$-differential privacy, and derives new upper bounds for minimum noise amplitude and noise power.
Section~\ref{sec:lower_bound} derives new lower bounds on the minimum noise magnitude and noise power.
Section~\ref{sec:tight} shows that the multiplicative gap between the lower bounds and the upper bounds goes to zero in various privacy regimes, and thus proves the tightness of the new lower and upper bounds.
Section~\ref{sec:comparison} conducts comprehensive numeric experiments to compare the performance of the truncated Laplacian mechanism with the optimal Gaussian mechanisms, and demonstrates the improvement in all privacy regimes.
Section~\ref{sec:conclusion} discusses some additional properties of the truncated Laplacian mechanism and concludes this paper.

\section{Problem Formulation} \label{sec:formulation}
 In this section, we first give some preliminaries on differential privacy, and then define the minimum noise amplitude $\opta$ and minimum noise power $\optb$ for $(\e,\delta)$-differentially private noise-adding mechanisms. 

Consider a real-valued query function $q: \database \rightarrow \R$,
where $\database$ is the set of all possible datasets. The real-valued query function $q$ will be applied to a dataset, and the query output is a real number. Two datasets $D_1, D_2 \in \database$ are called neighboring datasets if they differ in at most one element, i.e.,  one is a proper subset of the other and the larger dataset contains just one additional element \cite{DPsurvey}. A randomized query-answering mechanism $\KM$ for the query function $q$ will randomly output a number with probability distribution depending on query output $q(D)$, where $D$ is the dataset.

\begin{definition}[$(\e,\delta)$-differential privacy \citep{DKMMN06}]
A randomized mechanism $\KM$ gives $(\e, \delta)$-differential privacy if for all data sets $D_1$ and $D_2$ differing on at most one element, and for any measurable set $S \subset \text{Range}(\KM)$,
\begin{align}
    \text{Pr}[\KM(D_1) \in S] \le e^\e \;  \text{Pr}[\KM(D_2) \in S] + \delta. \label{eq:dp_eps_delta_constraint} 
 \end{align}
\end{definition}

The sensitivity of a real-valued query function measures how the query changes for neighboring datasets.
\begin{definition}[Query Sensitivity]
The sensitivity of $q$ is defined as
\begin{align*}
    \D := \max_{D_1,D_2 \in \database} | q(D_1) - q(D_2)|,
\end{align*}
for all $D_1,D_2$ differing in at most one element.
\end{definition}

A standard approach for preserving differential privacy is query-output independent noise-adding mechanisms, where a random noise is added to the query output.  Given a dataset $D$, a query-output independent noise-adding mechanism $\KM$ will release the query output $t = q(D)$ corrupted by an additive random noise $X$ with probability distribution $\p$:
\begin{align*}
    \KM(D) = t + X.
\end{align*}

We derive the differential privacy constraint on the noise probability distribution $\p$ in Lemma \ref{lem:dp_constraint}.
\begin{lemma}\label{lem:dp_constraint}{}
Given the query sensitivity $\D$ and privacy parameters $\e$ and $\delta$, the noise probability distribution $\p$ preserves $(\e,\delta)$-differential privacy if and only if
\begin{align}
  \p (S) - e^\e \p(S + d) \le \delta, \forall \; |d| \le \D, \text{measurable set} \; S \subset \R.  \label{eq:mainconstraint}
\end{align} 
\end{lemma}

\begin{proof}
The differential privacy constraint \eqref{eq:dp_eps_delta_constraint} on $\KM$ is that for any $t_1,t_2 \in \R$ such that $|t_1 - t_2| \le \D $ (corresponding to the query outputs for two neighboring datasets\footnote{In this work we impose no prior on the query function other than the query sensitivity $\D$. For any $t_1,t_2 \in \R$ such that $|t_1 - t_2| \le \D$, there may exist two neighboring datasets $D_1$ and $D_2$ with $q(D_1) = t_1$ and $q(D_2) = t_2$.}),
\begin{align}
    \p(S - t_1) \le e^\e \p(S - t_2) + \delta, \forall \; \text{measurable set} \; S \subset \R, \label{eq:tmp1}
\end{align}
where $\forall t \in \R$, $S+t$ is defined as the set $\{s+t \, | \, s \in S\}$.

Since \eqref{eq:tmp1} has to hold for any measurable set $S$ and any $|t_1 - t_2| \le \D$, equivalently, we have
\begin{align*}
    \p (S) \le e^\e \p(S + d) + \delta, \forall \; |d| \le \D, \text{measurable set} \; S \subset \R.
\end{align*}

\end{proof}

Let $\edp$ denote the set of noise probability distributions satisfying the $(\e, \delta)$-differential privacy constraint~\eqref{eq:mainconstraint}. Given $\p \in \edp$, the expected noise amplitude and noise power are $\int_{x \in \R} |x| \p(dx)$  and $\int_{x \in \R} x^2 \p(dx)$. The goal of this work is to characterize the minimum expected noise amplitude and noise power under $(\e,\delta)$-differential privacy. More precisely, define 
\begin{align*}
 & \opta&:= \inf_{\p \in \edp} \int_{x \in \R}  |x|  \p(dx) & \;\;\; (\text{min noise amplitude}), \\
 & \optb&:= \inf_{\p \in \edp} \int_{x \in \R}  x^2  \p(dx) & \;\;\; (\text{min noise power}).
\end{align*}

In this work, we characterize $\opta$ and $\optb$ in terms of $\D, \e, \delta$ by deriving tight lower bounds $\lowa, \lowb$ and upper bounds $\uppa, \uppb$ such that
$\lowa \le \opta \le \uppa$ and $\lowb  \le \optb \le \uppb.$

In the next section, we present the new upper bounds $\uppa$ and $\uppb$. The lower bounds $\lowa$ and $\lowb$ are presented in Section \ref{sec:lower_bound}.

\section{Upper Bound: Truncated Laplacian Mechanism} \label{sec:upper_bound}
In this section, we present a new class of $(\e, \delta)$-differentially private noise-adding mechanism, \emph{truncated Laplacian} mechanism. Applying the truncated Laplacian mechanism, we derive new achievable (and tight) upper bounds $\uppa$ and $\uppb$ on minimum noise amplitude $\opta$ and minimum noise power $\optb$ in Theorem \ref{thm:noise-amplitude} and Theorem \ref{thm:noise-power}.

Before presenting the exact form of the truncated Laplacian mechanism, we first discuss some key ideas and insights behind the new mechanism design. 

The standard Laplacian distribution for preserving $\e$-differential privacy has a symmetric probability density function $f(x) = \frac{\e}{2\D} e^{-\frac{|x|\e}{\D}}$. Note that for any $x \ge 0$, the probability density decay rate, $\frac{f(x)}{f(x+\D)}$, is exactly $e^{\e}$. \citet{GV_IT_Epsilon} show that the decay rate $e^{\e}$ is optimal under $\e$-differential privacy. Indeed, if the decay rate is higher, it is no longer $\e$-differentially private; if the decay rate is lower, it will incur a higher cost. However, under $(\e, \delta)$-differential privacy, Laplacian distribution is not optimal as it has a heavy tail distribution. 

$(\e, \delta)$-differential privacy relaxes the $\e$-differential privacy constraint, and it allows that for a set of points with a probability mass $\delta$, the decay rate can exceed $e^{\e}$. The Gaussian mechanism is widely used in $(\e, \delta)$-differential privacy, and for $x>0$, its probability density decay rate is $\frac{f(x)}{f(x+\D)} = \frac{e^{-\frac{x^2}{\sigma^2}}}{e^{-\frac{(x+\D)^2}{\sigma^2}}} = e^{\frac{\D^2 + 2\D x}{\sigma^2}} = e^{\frac{\D^2}{\sigma^2}} e^{\frac{2\D}{\sigma^2}x} $, which is exponentially increasing with respect to $x$. When $x$ is big, the decay rate can be very high. While the Gaussian mechanism addresses the long tail distribution to some extent by having higher decay rate for large $x$, the decay rate is smaller than $e^{\e}$ when $x$ is small.

Motivated by the observation that under $(\e, \delta)$-differential privacy, the decay rate shall be as high as possible without exceeding $e^{\e}$, except for a set of points with a probability mass $\delta$ (for those there is no limit on the decay rate), we derive a symmetric truncated Laplacian distribution where the probability density decay rate is exactly $e^{\e}$, except for a set of points with probability mass $\delta$ where the decay rate is infinite.
\begin{definition}[Truncated Laplacian Distribution]
Given the privacy parameters $0< \delta < \frac{1}{2}, \e > 0$ and the query sensitivity $\D > 0$, the probability density function of the truncated Laplacian distribution $\tlp$ is defined as:
\begin{align}
\ft(x) := 
    \begin{cases}
      B e^{-\frac{|x|}{\lambda}}, & \text{for } x \in [-A, A] \\
      0, & \text{otherwise}
    \end{cases} \label{eq:noise-pdf}
\end{align}

where
\begin{align*}
\lambda &:= \frac{\D}{\e}, \\
A &:= \frac{\D}{\e} \log(1 + \frac{e^\e - 1}{2\delta}), \\
\\
B &:= \frac{1}{2\lambda (1 - e^{-\frac{A}{\lambda}})} = \frac{1}{2\frac{\D}{\e} (1 - \frac{1}{ 1 + \frac{e^\e - 1}{2\delta}})}.
\end{align*}
\end{definition}

\begin{figure}[H]
\centering
\includegraphics[width=0.45\textwidth]{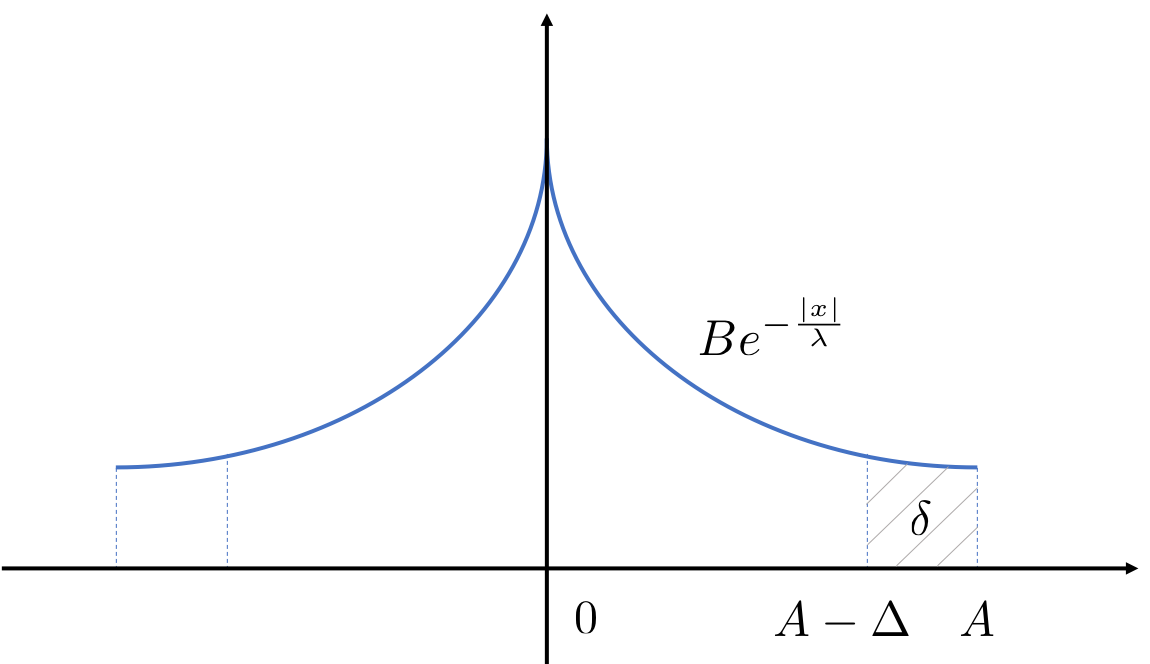}
\caption{Noise probability density function $\ft$ of the truncated Laplacian mechanism. $\ft$ is a symmetric truncated exponentinal function with a probability mass $\delta$ in the last interval with length $\D$ in the support of $\ft$, i.e., the interval $[A-\D, A]$. The decay rate $\frac{\ft(x)}{\ft(x+\D)}$ is exactly $e^\e$ for $x \in [0, A-\D)$. The parameters $A$ and $B$ are then derived by solving the equations that $\int_{x \in \R} \ft(x)dx = 1$ and $\int_{A-\D}^{A} \ft(x) dx = \delta$.}
\label{fig:density}
\end{figure}

$\ft$ is a valid probability density function, as $\ft(x) \ge 0$ and $\int_{x \in \R} \ft(x)dx = \int_{0}^{A} 2 B e^{-\frac{|x|}{\lambda}} dx
= 2 \lambda B (1 - e^{-\frac{A}{\lambda}}) = 1$.

We discuss the key properties of the symmetric probability density function $\ft(x)$:
\begin{itemize}
    \item The decay rate in $[0, A-\D]$ is exactly $e^{\e}$, i.e., $\frac{\ft(x)}{\ft(x+\D)} = e^{\e} , \; \forall x \in [0, A-\D]$.
    \item The probability mass in the interval $[A-\D, A]$ is $\delta$, i.e., $\tlp([A - \D, A]) = \delta$. Indeed, 
        \begin{align*}
        &\int_{A-\D}^{A} \ft(x) dx = \int_{A-\D}^{A} B e^{-\frac{|x|}{\lambda}} dx \\
        & = \lambda B (e^{-\frac{A-\D}{\lambda}} - e^{-\frac{A}{\lambda}}) = \lambda B e^{-\frac{A}{\lambda}} (e^{\frac{\D}{\lambda}} - 1) \\
        &= (e^{\frac{\D}{\lambda}} - 1) \frac{e^{-\frac{A}{\lambda}}}{2(1-e^{-\frac{A}{\lambda}})} = (e^{\frac{\D}{\lambda}} - 1) \frac{1}{2 (e^{\frac{A}{\lambda}}-1)} = \delta.
        \end{align*}
    \item The decay rate $\frac{\ft(x)}{\ft(x+\D)}$ is $+\infty$ for $ x \in (A-\D, A]$, as  $\ft(x) = 0$ for  $x \in (A, +\infty)$.
\end{itemize}

\begin{definition}[Truncated Laplacian mechanism]
Given the query sensitivity $\D$, and the privacy parameters $\e, \delta$, the truncated Laplacian mechanism adds a noise with probability distribution $\tlp$ defined in \eqref{eq:noise-pdf} to the query output. 
\end{definition}

\begin{theorem}
The truncated Laplacian mechanism preserves $(\e,\delta)$-differential privacy.
\end{theorem}

\begin{proof}
Equivalently, we need to show that the truncated Laplacian distribution $\tlp$ defined in \eqref{eq:noise-pdf} satisfies the $(\e,\delta)$-differential privacy constraint \eqref{eq:mainconstraint}. 

We are interested in maximizing $\tlp (S)  - e^\e \tlp(S + d)$ in \eqref{eq:mainconstraint} and show that the maximum over $S \subseteq \R$ is upper bounded by $\delta$. Since $\ft(x)$ is symmetric and monotonically decreasing in $[0, +\infty)$, without loss of generality, we can assume $d \ge 0$ and thus $d \in [0, \D]$.

To maximize $\tlp (S)  - e^\e \tlp(S + d)$, $S$ shall not contain points in $(-\infty, -\frac{\D}{2}]$, as
\begin{align*}
\ft(x) \le \ft(x+d), \forall x \in (-\infty, -\frac{\D}{2}].
\end{align*}

$S$ shall not contain points in $[-\frac{\D}{2}, A - \D]$, as
\begin{align*}
\ft(x) \le e^\e \ft(x+d), \forall x \in [-\frac{\D}{2}, A - \D].
\end{align*}

Therefore, $\tlp (S)  - e^\e \tlp(S + d)$ is maximized for some set $S \subseteq [A -\D, +\infty)$. Since $\ft(x)$ is monotonically decreasing in  $[A -\D, +\infty)$, $\tlp (S)  - e^\e \tlp(S + d)$ is maximized at $S = [A - \D, +\infty)$ and the maximum value is $\int_{A-\D}^{A-\D + d} f(x)dx \le \int_{A-\D}^{A} \ft(x) dx = \delta.$

We conclude that $\tlp$ satisfies the $(\e,\delta)$-differential privacy constraint \eqref{eq:mainconstraint}.
\end{proof}

Next, we apply the truncated Laplacian mechanism to derive new upper bounds on the minimum noise amplitude $\opta$ and noise power $\optb$.

\begin{theorem}[Upper Bound on Minimum Noise Amplitude]\label{thm:noise-amplitude}
\begin{align}
\opta \le \uppa:= \frac{\D}{\e} (1- \frac{\log(1 + \frac{e^\e - 1}{2\delta})}{\frac{e^\e - 1}{2\delta}}). \label{eq:l1-upperbound}
\end{align}
\end{theorem}

\begin{proof}
We can compute the expected noise amplitude for the truncated Laplacian distribution $\tlp$ defined in \eqref{eq:noise-pdf} via 
\begin{align*}
\uppa& := \int_{x \in \R} \ft(x) |x| dx  = 2  \int_0^A B e^{-\frac{x}{\lambda}} x dx   \\
&= 2B \lambda ( - A e^{-\frac{A}{\lambda}} + \int_0^A e^{-\frac{x}{\lambda}} dx ) \\
&= 2B \lambda ( - A e^{-\frac{A}{\lambda}} + \lambda(1- e^{-\frac{A}{\lambda}}) )  \\
&= -\frac{A e^{-\frac{A}{\lambda}}}{1- e^{-\frac{A}{\lambda}}}  + \lambda  = \lambda - \frac{A}{e^{\frac{A}{\lambda}}-1}  \\
& = \frac{\D}{\e} (1- \frac{\log(1 + \frac{e^\e - 1}{2\delta})}{\frac{e^\e - 1}{2\delta}}).
\end{align*}

Since the truncated Laplacian mechanism preserves $(\e, \delta)$-differential privacy, this gives an upper bound on the minimum noise amplitude $\opta$ under $(\e, \delta)$-differential privacy.
\end{proof}

In Theorem \ref{thm:noise-amplitude}, the upper bound $\uppa$ is composed of two parts. The first part is $\frac{\D}{\e}$, which is the noise amplitude of the Laplacian mechanism under $\e$-differential privacy. The second part reduces the noise by a portion of $\frac{\log(1 + \frac{e^\e - 1}{2\delta})}{\frac{e^\e - 1}{2\delta}}$ due to the $\delta$-relaxation in $(\e, \delta)$-differential privacy.

We analyze the asympotic properties of $\uppa$ in the high privacy regimes as $\e \to 0, \delta \to 0$:
\begin{itemize}
    \item Given $\e$, $\lim_{\delta \to 0} \uppa = \frac{\D}{\e}$.  The truncated Laplacian mechanism will be reduced to the standard Laplacian mechanism as $\delta \to 0$,.
    
    \item Given $\delta$, $\lim_{\e \to 0} \uppa = \frac{\D}{4\delta}$. Indeed, when $\e \to 0$, $\frac{e^\e-1}{2\delta} \to 0$, and thus
    \begin{align*}
      \uppa &\approx \frac{\D}{\e}   (1- \frac{   \frac{e^\e - 1}{2\delta} - \frac{(\frac{e^\e - 1}{2\delta})^2}{2}   }{\frac{e^\e - 1}{2\delta}} ) \\
      &=  \frac{\D}{\e} \frac{\frac{e^\e - 1}{2\delta}}{2} \approx \frac{\D}{\e} \frac{\e}{4\delta} = \frac{\D}{4\delta}.
    \end{align*}
    As $\e \to 0$, the truncated Laplacian distribution is reduced to a uniform distribution in the interval $[-\frac{\D}{2\delta}, \frac{\D}{2\delta}]$ with probability density $\frac{\delta}{\D}$.
    
    \item In the regime $\delta = \e \to 0$, the upper bound
    \begin{align}
     \uppa &= \frac{\D}{\e} (1- \frac{\log(1 + \frac{e^\e - 1}{2\e})}{\frac{e^\e - 1}{2\e}}) \nonumber \\
     & \approx \frac{\D}{\e} (1- \frac{\log(1 + \frac{\e}{2\e})}{\frac{\e}{2\e}})  \nonumber \\
     &= \frac{\D}{\e} (1- 2 \log{\frac{3}{2}}). \label{eq:truncated-lap-upper-l1}
    \end{align}
    In Section \ref{sec:tight}, we show that the constant factor $(1- 2 \log{\frac{3}{2}})$ is actually tight and the upper bound $\uppa$ matches the lower bound $\lowa$ defined in Theorem \ref{thm:noise-amplitude-lower-bound}.
\end{itemize}

\begin{theorem}[Upper Bound on Minimum Noise Power]\label{thm:noise-power}
Define 
  \begin{align}
  \uppb:= \frac{2\D^2}{\e^2} (1 -  \frac{ \frac{1}{2}\log^2 (1+\frac{e^\e - 1}{2\delta}) + \log(1+\frac{e^\e - 1}{2\delta})  }{\frac{e^\e - 1}{2\delta}}). \label{eq:l2-upperbound}
  \end{align}
  We have 
  \begin{align*}
  \optb \le \uppb.
  \end{align*}
  \end{theorem}
  
\begin{proof}
We can compute the cost for the truncated Laplacian distribution via 
  \begin{align*}
  \uppb & := \int_{x \in \R} \ft(x) x^2 dx = 2 \int_0^A f(x) x^2 dx = 2  \int_0^A B e^{-\frac{x}{\lambda}} x^2 dx   \\
      &= 2B \lambda ( -A^2 e^{-\frac{A}{\lambda}} + \int_0^A e^{-\frac{x}{\lambda}} 2x dx )  \\
      &= 2B \lambda (-A^2 e^{-\frac{A}{\lambda}} + 2\lambda( - A e^{-\frac{A}{\lambda}} +  \int_0^A e^{-\frac{x}{\lambda}}dx) ) \\
      &= 2B \lambda (-A^2 e^{-\frac{A}{\lambda}} + 2\lambda( - A e^{-\frac{A}{\lambda}} +  \lambda - \lambda e^{-\frac{A}{\lambda}} ) \\
      &= \frac{ -A^2 e^{-\frac{A}{\lambda}}  -  2\lambda A e^{-\frac{A}{\lambda}} + 2 \lambda^2 - 2\lambda^2 e^{-\frac{A}{\lambda}}}{1- e^{-\frac{A}{\lambda}}}  \\
      &= 2 \lambda^2 - \frac{A^2 e^{-\frac{A}{\lambda}}  +  2\lambda A e^{-\frac{A}{\lambda}}}{1- e^{-\frac{A}{\lambda}}}  \\
      &= 2 \lambda^2 - \frac{A^2   +  2\lambda A}{e^{\frac{A}{\lambda}}-1}  \\
      &=\frac{2\D^2}{\e^2} - \frac{\frac{\D^2}{\e^2} \log^2 (1+\frac{e^\e - 1}{2\delta}) + \frac{2\D^2}{\e^2} \log(1+\frac{e^\e - 1}{2\delta})  }{\frac{e^\e - 1}{2\delta}} \\
      &= \frac{2\D^2}{\e^2} (1 -  \frac{ \frac{1}{2}\log^2 (1+\frac{e^\e - 1}{2\delta}) + \log(1+\frac{e^\e - 1}{2\delta})  }{\frac{e^\e - 1}{2\delta}}).
  \end{align*}

Since $\ft(x)$ can preserve $(\e, \delta)$-differential privacy, this gives an upper bound on  $\optb$.
 
\end{proof}


It turns out that the upper bounds $\uppa$ and $\uppb$ in Theorem \ref{thm:noise-amplitude} and Theorem \ref{thm:noise-power} are tight. We derive new lower bounds for $\opta$ and $\optb$ in the next section, and show that the multiplicative gap between the lower bounds and the upper bounds goes to zero in the high privacy regions in Section \ref{sec:tight}.

\section{Lower Bound} \label{sec:lower_bound}
In this section, we derive new lower bounds $\lowa$ and $\lowb$ on the minimum noise amplitude $\opta$ and minimum noise power $\optb$, respectively. The key technique is to discretize the continuous probability distribution and the loss function, and transform the continuous functional optimization problem to linear programming, and then apply the discrete result from \citet{GV_IT_Approximate}.

\citet{GV_IT_Approximate} derived lower bounds for an \emph{integer-valued} query function under $(\e, \delta)$-differential privacy. For integer-valued query functions, they formulate a linear programming problem with the objective of minimizing the additive noise. They studied the dual problem and constructed a dual feasible solution which gives a lower bound. Extending this result to the continuous setting, we show a similar lower bound for \emph{real-valued} query function under $(\e, \delta)$-differential privacy. 


First, we give a lower bound for $(\e, \delta)$-differential privacy for \emph{integer-valued} query function due to \citet{GV_IT_Approximate}.

Define
\begin{align*}
    a &:= \frac{\delta + \frac{e^{\e}-1}{2}}{e^\e},\\
    b &:= e^{-\e}.
\end{align*}
To avoid integer rounding issues, assume that there exists an integer $n$ such that $\sum_{k=0}^{n-1} a b^k = \frac{1}{2}$.

\begin{lemma}[Theorem 8 in \citet{GV_IT_Approximate}]\label{lem:discrete_dp}
Consider a symmetric cost function $\loss(\cdot): \Z \rightarrow \R$, where $\Z$ denotes the set of all integers. Given the privacy parameters $\e, \delta$ and the discrete query sensitivity $\dD \in \Z^+$, if a discrete probability distribution $\p$ satisfies
\begin{align}
    \p(S) - e^\e \p(S+d) \le \delta, \forall S \subseteq \Z, \forall d \in \Z, |d| \le \dD \label{eq:discrete_dp_constraint}
\end{align}
and the cost function $\loss(\cdot)$ satisfies
\begin{align}
  \sum_{i=1}^{n-1} b^i \big(2 \loss(i\dD) - \loss(1+(i-1)\dD) - \loss(1 + i\dD) \big) \ge \loss(1), \label{eq:gv_condition}
\end{align}
then we have
\begin{align}
    \Sigma_{i \in \Z} \loss(i) \p(i) \ge  2 \sum_{k=0}^{n-1} a b^k \loss(1 + k\dD). \label{eq:gv_lowerbound}{}
  \end{align}
\end{lemma}

\begin{theorem}[Lower Bound on Minimum Noise Amplitude] \label{thm:noise-amplitude-lower-bound}
Define
  \begin{align}
 \lowa&:= 2 a \sum_{k=0}^{n-1} b^k k \D \nonumber \\
  &= 2 a \left(\frac{b-b^n}{(1-b)^2}  - \frac{(n-1)b^n}{1-b}\right) \D. \label{eq:l1-lower-bound}
  \end{align}
 We have
  \begin{align*}
  \opta \ge \lowa.
  \end{align*}
\end{theorem}

\begin{proof}
Given $\p \in \edp$, we can derive a lower bound on the cost by discretizing the probability distributions and applying the lower bound \eqref{eq:gv_lowerbound} for integer-valued query functions in Lemma \ref{lem:discrete_dp}.

We first discretize the probability distributions $\p$. Given a positive integer $N \ge 0$, define a discrete probability distribution $\dpp$ via
\begin{align*}
  \dpp(i) := \p\big([\frac{\D}{2N}(2i-1), \frac{\D}{2N}(2i+1))\big), \forall i \in \Z.
\end{align*}

For the noise cost function $|x|$, define the corresponding discrete cost function $\dll$ via
\begin{align*}
  \dll(i) \triangleq  \begin{cases}
      0, & i = 0 \\
      \frac{\D}{2N}(2i-1), & i \ge 1 \\
      \dll(-i), & i < 0.
    \end{cases}
\end{align*}

It is ready to see that
\begin{align*}
  \int_{x \in \R} |x|  \p(dx) \ge \Sigma_{i \in \Z} \dpp(i) \dll(i).
\end{align*}

As the continuous probability distribution $\p$ satisfies $(\e,\delta)$-differential privacy constraint  \eqref{eq:mainconstraint} with the query sensitivity $\D$, the discrete probability distribution $\dpp$ satisfies the discrete $(\e, \delta)$-differential privacy constraint \eqref{eq:discrete_dp_constraint} with query sensitivity $\dD = N$, i.e., $\dpp$ satisfies
\begin{align*}
  \dpp (S)  - e^\e \dpp(S + d) \le  \delta, \forall S \subseteq \Z, |d| \le N.
\end{align*}

We can verify that the condition \eqref{eq:gv_condition} in Lemma \ref{lem:discrete_dp} holds for $\dll$ and $\dpp$ with query sensitivity $\dD = N$ when $N$ is sufficiently large. Indeed, when $N \ge a + 2$,
\begin{align*}
 &\sum_{i=1}^{n-1} b^i [2 \dll(iN) - \dll(1+(i-1)N) - \dll(1 + iN) ] \\
 & \;\;\;\;\;- \dll(1) \\
  &= \sum_{i=1}^{n-1} b^i \frac{\D}{2N} [2(2iN -1) - 2 (1 + (i-1)N) + 1 \\
  & \;\;\;\;\; - 2(1+iN) + 1 ] - \frac{\D}{2N} \\
  &= \sum_{i=1}^{n-1} b^i \frac{\D}{2N} (2N - 4) - \frac{\D}{2N} \\
  &= \frac{\D}{2N} (\sum_{i=1}^{n-1} b^i ( 2N - 4) - 1) \\
  &= \frac{\D}{2N} (\frac{2N-4}{2a} - 1) = \frac{\D}{2N} (\frac{N-2}{a} - 1)  \ge 0.
\end{align*}

 The corresponding lower bound in \eqref{eq:gv_lowerbound} for $\dll$ and $\dpp$ is
\begin{align*}
  & 2 \sum_{k=0}^{n-1} ab^k \dll(1 + kN) = 2 \sum_{k=0}^{n-1} ab^k \frac{\D}{2N}(2kN + 1) \\
  & = 2 \sum_{k=0}^{n-1} ab^k ( k\D  + \frac{\D}{2N}) = 2 a\D \sum_{k=0}^{n-1} b^k k + \frac{\D}{2N} \\
  &\ge 2 a\D \sum_{k=0}^{n-1} b^k k = 2 a \left(\frac{b-b^n}{(1-b)^2}  - \frac{(n-1)b^n}{1-b}\right) \D \\
  &= \lowa
\end{align*}

Therefore, for any $\p \in \edp$, we have
\begin{align*}
  \int_{x \in \R}  |x|  \p(dx) \ge \Sigma_{i \in \Z} \dpp(i) \dll(i) \ge \lowa,
\end{align*}
and thus $\opta \ge \lowa.$
\end{proof}

Similarly, we derive the lower bound for the minimum noise power $\optb$.

\begin{theorem}[Lower Bound on Minimum Noise Power] \label{thm:noise-power-lower-bound}
Define 
\begin{align}
  \lowb &:= 2 \sum_{k=0}^{n-1} ab^k  k^2 \D^2 \nonumber \\
  &=  \frac{2a \D^2}{1-b} [-b + 2 ( \frac{b(1-b^{n-1})}{(1-b)^2} - \frac{(n-1) b^{n}}{1-b} )\nonumber  \\
  & \;\;\;\;- \frac{b^2(1-b^{n-2})}{1-b} - (n-1)^2 b^{n}] . \label{eq:l2-lower-bound}
\end{align}
We have
  \begin{align*}
  & \optb \ge \lowb.
  \end{align*}
\end{theorem}

\begin{proof}
We first discretize the probability distribution $\p$. Given a positive integer $N \ge 0$, define a discrete probability distribution $\dpp$ via
\begin{align*}
  \dpp(i) \triangleq \p\big([\frac{\D}{2N}(2i-1), \frac{\D}{2N}(2i+1))\big), \forall i \in \Z.
\end{align*}

Define the corresponding discrete cost function $\dll$ via
\begin{align*}
  \dll(i) \triangleq  \begin{cases}
      0, & i = 0 \\
      (\frac{\D}{2N}(2i-1))^2, & i \ge 1 \\
      \dll(-i), & i < 0.
    \end{cases}
\end{align*}

It is easy to see that
\begin{align*}
  \int_{x \in \R} \loss (x)  \p(dx) \ge \Sigma_{i \in \Z} \dpp(i) \dll(i).
\end{align*}

As the continuous probability distribution $\p$ satisfies $(\e,\delta)$-differential privacy constraint with continuous query sensitivity $\D$, the discrete probability distribution $\dpp$ satisfies the discrete $(\e, \delta)$-differential privacy constraint with discrete query sensitivity $N$, i.e., $\dpp$ satisfies
\begin{align*}
  \dpp (S)  - e^\e \dpp(S + d) \le  \delta, \forall S \subseteq \Z, |d| \le N.
\end{align*}

Next we verify that the condition (9) in Lemma 2 holds when $N$ is sufficiently large for the $\ell^2$ cost function. Indeed,
\begin{align*}
 &\;\; \sum_{i=1}^{n-1} b^i \big(2 \dll(iN) - \dll(1+(i-1)N) - \dll(1 + iN) \big) - \dll(1) \\
  &= \sum_{i=1}^{n-1} b^i \frac{\D^2}{4N^2} \big(2(2iN -1)^2 - (2 (1 + (i-1)N) - 1)^2  - (2(1+iN) - 1)^2 \big) - \frac{\D^2}{4N^2} \\
  &= \sum_{i=1}^{n-1} b^i \frac{\D^2}{4N^2} ((8i - 4)N^2 - 16 i N + 4N) - \frac{\D^2}{4N^2} \\
  & \ge 0,
\end{align*}
 where the last step holds when $N$ is sufficiently large.

 The lower bound in (10) is
\begin{align*}
  & 2 \sum_{k=0}^{n-1} ab^k \dll(1 + kN) \\
  & = 2 \sum_{k=0}^{n-1} ab^k \frac{\D^2}{4N^2}(2kN + 1)^2 \\
  & = 2 \sum_{k=0}^{n-1} ab^k \frac{\D^2}{4N^2} ( 4k^2 N^2  +  4kN + 1) \\
  & \ge 2 \sum_{k=0}^{n-1} ab^k \frac{\D^2}{4N^2} 4k^2 N^2  \\
  &=2 \sum_{k=0}^{n-1} ab^k  k^2 \D^2 \\
  &= 2a  \frac{-b + 2 ( \frac{b(1-b^{n-1})}{(1-b)^2} - \frac{(n-1) b^{n}}{1-b} ) - \frac{b^2(1-b^{n-2})}{1-b} - (n-1)^2 b^{n}}{1-b}  \D^2.
\end{align*}

Therefore, for any $\p \in \edp$, we have
\begin{align*}
  \int_{x \in \R}  x^2  \p(dx) & \ge \Sigma_{i \in \Z} \dpp(i) \dll(i) \\
  & \ge  2a  \frac{-b + 2 ( \frac{b(1-b^{n-1})}{(1-b)^2} - \frac{(n-1) b^{n}}{1-b} ) - \frac{b^2(1-b^{n-2})}{1-b} - (n-1)^2 b^{n}}{1-b}  \D^2,
\end{align*}
and thus
\begin{align*}
  \optb \ge   2a  \frac{-b + 2 ( \frac{b(1-b^{n-1})}{(1-b)^2} - \frac{(n-1) b^{n}}{1-b} ) - \frac{b^2(1-b^{n-2})}{1-b} - (n-1)^2 b^{n}}{1-b}  \D^2 = \lowb.
  \end{align*} 

\end{proof}

\section{Tightness of the Lower and Upper Bounds} \label{sec:tight}
In this section, we compare the lower bounds $\lowa, \lowb$ and the upper bounds $\uppa, \uppb$ (derived from the truncated Laplacian mechanism) for the minimum noise amplitude and noise power under $(\e,\delta)$-differential privacy. We show that they are close in the high privacy regions and the multiplicative gap goes to zero, which proves the tightness of these lower and upper bounds and thus establishes the near-optimality of the truncated Laplacian mechanism.  

\begin{theorem}[Tightness of Lower bound and Upper bound on Minimum Noise Amplitude]\label{thm:tight_l1}
\begin{align*}
  \lim_{\e \to 0}\frac{\lowa}{\uppa} &\ge 1 - 2\delta. \\
  \lim_{\delta \to 0} \frac{\lowa}{\uppa} &\ge  \frac{\e}{e^\e - 1} = 1 - 
\frac{\e}{2} + O(\e^2). \\
  \lim_{\e = \delta \to 0} \frac{\lowa}{\uppa} &= 1.
\end{align*}
\end{theorem}

\begin{proof}

1. $\delta$ is fixed, and $\e \to 0$: 

When $\e \to 0$, the upper bound $\uppa \to \frac{\D}{4\delta}$, and the lower bound  $\lowa \to 2 \delta \frac{n(n-1)}{2} \D = 2 \delta \frac{\frac{1}{2\delta}(\frac{1}{2\delta}-1)}{2} \D = (\frac{1}{4\delta} - \frac{1}{2})\D$. Therefore, 
  \begin{align*}
    \lim_{\e \to 0}\frac{\lowa}{\uppa} \ge \frac{ (\frac{1}{4\delta} - \frac{1}{2})\D}{\frac{\D}{4\delta}} = 1 - 2\delta.
\end{align*}

Note that $1 - 2\delta \to 1$, as $\delta \to 0$, and thus the multiplicative gap between $\lowa$ and $\uppa$ converges to zero.

2. $\e$ is fixed, and $\delta \to 0$: 

When $\delta \to 0$, the upper bound $\uppa \to \frac{\D}{\e}$. For the lower bound  $\lowa$, we have
  \begin{align*}
    a & \to \frac{1- e^{-\e}}{2}, \\
    b^n &\to 0, \\
    nb^n &\to 0,
  \end{align*}
  and thus  $\lowa \to \frac{\D}{\e^\e -1} $ as $\delta \to 0$. Therefore,
  \begin{align*}
  \lim_{\delta \to 0} \frac{\lowa}{\uppa} \ge \frac{\frac{\D}{\e^\e -1}}{\frac{\D}{\e}} =  \frac{\e}{e^\e - 1} = 1 - \frac{\e}{2} + O(\e^2).
\end{align*}

Therefore, the multiplicative gap between $\lowa$ and $\uppa$ converges to zero as $\e \to 0$.

3. $\epsilon = \delta \to 0$:

In this regime, $\uppa \approx \frac{\D}{\e} (1- 2 \log{\frac{3}{2}})$ as shown in Section~\ref{sec:upper_bound}. For the lower bound $\lowa$, since $\sum_{k=0}^{n-1}a b^k = \frac{1}{2}$, we have
\begin{align*}
  a\frac{1-b^n}{1-b} = \frac{1}{2} \Rightarrow b^n = 1 - \frac{1-b}{2a}.
\end{align*}

As $\e=\delta \to 0$,  $\frac{1-b}{2a} = \frac{1- e^{-\e}}{2 \frac{\delta + \frac{e^\e - 1}{2}}{e^\e}} \to \frac{1}{3}$, and thus
\begin{align*}
  \lim_{\delta \to 0} b^n &= 1 - \frac{1}{3} = \frac{2}{3}, \\
  n &= \Theta(\frac{\log \frac{3}{2}}{\delta}).
\end{align*}

Note that $a = \Theta(\frac{3}{2} \delta) $ as $\delta \to 0$.

Therefore, as $\e = \delta \to 0$,
\begin{align*}
  & 2 a \left(\frac{b-b^n}{(1-b)^2}  - \frac{(n-1)b^n}{1-b}\right) \D \\
  &\approx 2  a (\frac{1 -  \frac{2}{3} }{\delta^2} - \frac{\frac{\log \frac{3}{2}}{\delta}\frac{2}{3} }{\delta}) \D \\
   &= 2 a ( \frac{1}{3\delta^2} - \frac{\frac{2}{3} \log (\frac{3}{2})}{\delta^2}) \D\\
  &\approx 2   \frac{3}{2} \delta ( \frac{1}{3\delta^2} - \frac{\frac{2}{3} \log (\frac{3}{2})}{\delta^2}) \D\\
  &= (1 - 2 \log \frac{3}{2}) \frac{\D}{\delta}.
\end{align*}

Therefore, $\opta$ is lower bounded by $\lowa \approx (1 - 2 \log \frac{3}{2}) \frac{\D}{\delta}$ in the regime $\e = \delta \to 0$. Since it is also upper bounded by $\uppa \approx \frac{\D}{\e} (1- 2 \log{\frac{3}{2}})$, we conclude that $\lim_{\e = \delta \to 0} \frac{\lowa}{\uppa} = 1$. 

Note that our result closes the constant multiplicative gap in the discrete setting (see Equation (67) and (69) in \citet{GV_IT_Approximate}).
\end{proof}

Similarly, we show that the lower bound $\lowb$ and the upper bound $\uppb$ on the minimum noise power are also tight.
\begin{theorem}[Tightness of Lower bound and Upper bound on Minimum Noise Power]\label{thm:tight_l2}
  \begin{align*}
    \lim_{\e \to 0} \frac{\lowb}{\uppb} &\ge (1-\delta)(1 - 2\delta) = 1 - 3\delta + 2\delta^2. \\
    \lim_{\delta \to 0} \frac{\lowb}{\uppb} &\ge   \frac{\e^ 2(1 + e^{\e})}{2(e^\e - 1)^2} = 1 - \frac{\e}{2} + O(\e^2). \\ 
    \lim_{\e = \delta \to 0} \frac{\lowb}{\uppb} &= 1.
  \end{align*}
  \end{theorem}
  
\begin{proof}
\begin{itemize}
    \item Case $\e \to 0$: When $\e \to 0$, the upper bound $\uppb$ converges to $\frac{\D^2}{12\delta^2}$. For the lower bound, when $\e \to 0$, we have 
    \begin{align*}
      a &\to \delta, \\
      b &\to 1, \\
      n &\to \frac{1}{2\delta},
    \end{align*}
    and thus the lower bound  $ \lowb = 2 a \sum_{k=0}^{n-1} b^k k^2 \D^2$ converges to 
    \[2 \delta \frac{(n-1)n(2n-1)}{6} \D^2 = 2 \delta \frac{ (\frac{1}{2\delta}-1)\frac{1}{2\delta}(\frac{1}{\delta}-1)}{6} \D^2 = \frac{\D^2}{12} (\frac{1}{\delta} - 1) (\frac{1}{\delta} - 2),\] which matches the upper bound as $\delta \to 0$. Therefore, 
    \begin{align*}
      \lim_{\e \to 0} \frac{\lowb}{\uppb} \ge \frac{ \frac{\D^2}{12} (\frac{1}{\delta} - 1) (\frac{1}{\delta} - 2)}{\frac{\D^2}{12\delta}} =  (1-\delta)(1 - 2\delta).
  \end{align*}

  \item Case $\delta \to 0$: When $\delta \to 0$, the upper bound $\uppb$ converges to $\frac{2\D^2}{\e^2}$. For the lower bound, we have
  \begin{align*}
    a & \to \frac{1- e^{-\e}}{2} = \frac{1-b}{2} \\
    b^n &\to 0 \\
    n^2b^n &\to 0,
  \end{align*}
  and thus the lower bound  $ 2a  \frac{-b + 2 ( \frac{b(1-b^{n-1})}{(1-b)^2} - \frac{(n-1) b^{n}}{1-b} ) - \frac{b^2(1-b^{n-2})}{1-b} - (n-1)^2 b^{n}}{1-b}  \D^2$ converges to 
  \begin{align*}
   (- b + 2 \frac{b}{(1-b)^2} - \frac{b^2}{1-b} ) \D^2 = \frac{b^2 + b}{(1-b)^2}\D^2 =  \frac{e^{-2\e} + e^{-\e}}{(1-e^{-\e})^2}\D^2 =  \frac{1 + e^{\e}}{(e^\e - 1)^2}\D^2,
  \end{align*}
  and this matches $\frac{2\D^2}{\e^2}$ as $\e \to 0$. Therefore,
  \begin{align*}
  \lim_{\delta \to 0} \frac{\lowb}{\uppb} \ge \frac{ \frac{1 + e^{\e}}{(e^\e - 1)^2}\D^2}{\frac{2\D^2}{\e^2}} =  \frac{\e^ 2(1 + e^{\e})}{2(e^\e - 1)^2}.
\end{align*}

\item Case $\epsilon = \delta \to 0$:  The upper bound $\uppb$ converges to $\frac{2\D^2}{\e^2} (1 - \log^2{\frac{3}{2}} - 2\log{\frac{3}{2}})$.
 When $\epsilon = \delta \to 0$,
the lower bound $\lowb$ is
  \begin{align*}
    & 2a  \frac{-b + 2 ( \frac{b(1-b^{n-1})}{(1-b)^2} - \frac{(n-1) b^{n}}{1-b} ) - \frac{b^2(1-b^{n-2})}{1-b} - (n-1)^2 b^{n}}{1-b}  \D^2 \\
    &\approx 2 \frac{3}{2}\e \frac{  2(\frac{1-\frac{2}{3}}{\e^2} - \frac{\frac{2}{3}\log\frac{3}{2}}{\e^2} )  -  \frac{1}{3\e}  - \frac{2}{3} \frac{\log^2\frac{3}{2}}{\e^2} }{\e} \D^2  \\
    &= \frac{3\D^2}{\e^2} (\frac{2}{3} - \frac{4}{3} \log\frac{3}{2}  -   \frac{2}{3} \log^2\frac{3}{2})\\
     &= \frac{2\D^2}{\e^2} (1 - 2 \log\frac{3}{2}  -  \log^2\frac{3}{2} ),
  \end{align*}
which matches the uppber bound. We conclude that
  \begin{align*}
    \lim_{\e = \delta \to 0} \frac{\lowb}{\uppb} = 1.
  \end{align*}
\end{itemize}
\end{proof}

\section{Comparison with the Optimal Gaussian Mechanism} \label{sec:comparison}
In this section we conduct numeric experiments to compare the performance of the truncated Laplacian mechanisms with the optimal Gaussian mechanism described in \citet{icmlGaussian}.

A classic result on the Gaussian mechanism is that for any $\e, \delta \in (0,1)$, adding a Gaussian noise with standard deviation $\sigma = \frac{\sqrt{2 \log (1.25/\delta)}}{\e}\D$ preserves $(\e, \delta)$-differential privacy \cite{DPbook}.  \citet{icmlGaussian} developed the optimal Gaussian mechanism whose variance is calibrated directly using the Gaussian cumulative density function instead of a tail bound approximation.

\begin{figure}[H]
    \centering
    \includegraphics[width=0.45\textwidth]{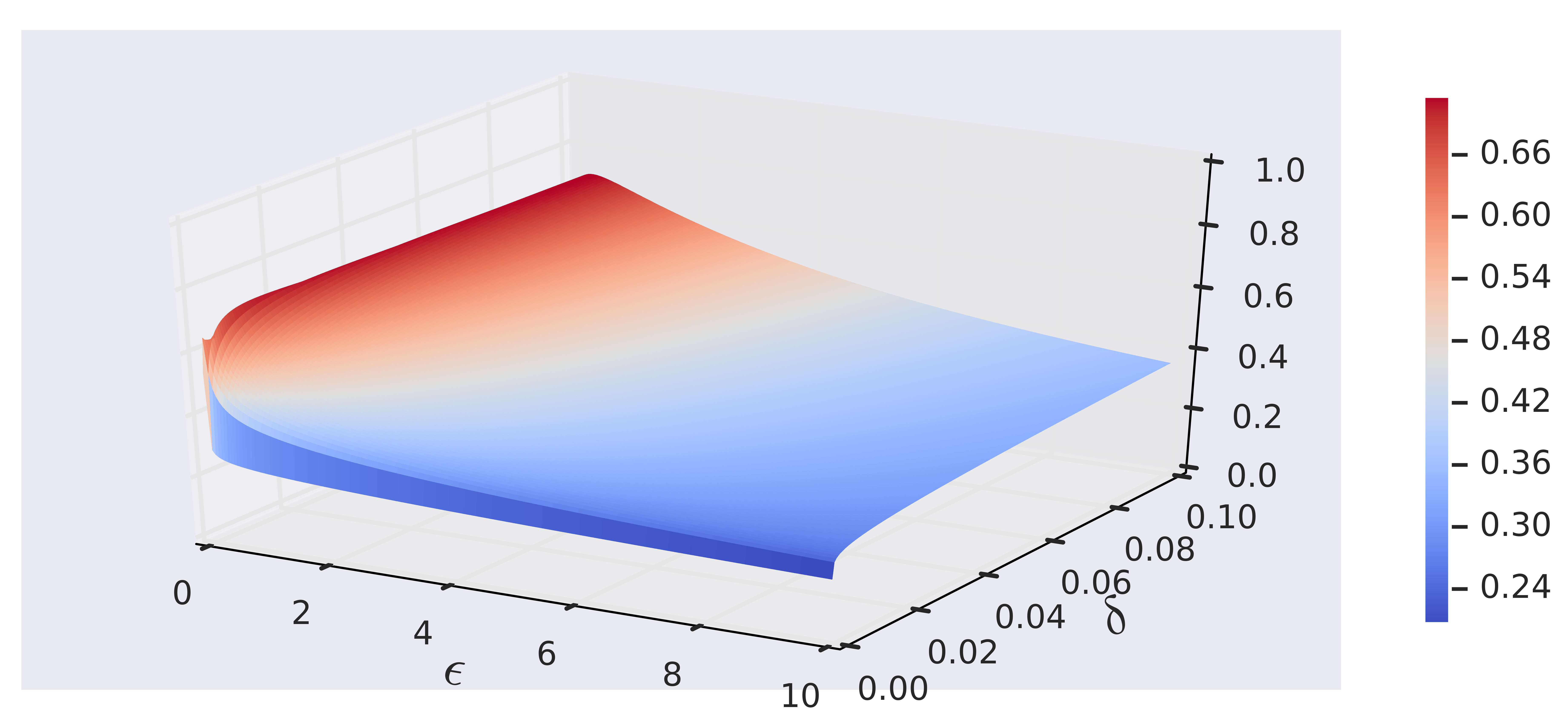}
    \caption{Ratio of the Noise Amplitude of the Truncated Laplacian Mechanism and the Optimal Gaussian Mechanism.}
    \label{fig:ratio_l1}
\end{figure}

We plot the ratio of the noise amplitude of truncated Laplacian mechanism and the optimal Gaussian mechanism in Fig.~\ref{fig:ratio_l1}, and plot the ratio of the noise power of truncated Laplacian mechanism and the optimal Gaussian mechanism in Fig.~\ref{fig:ratio_l2}, where $\e \in [10^{-4}, 10]$ and $\delta \in [10^{-6}, 0.1]$. Note that compared with the optimal Gaussian mechanism, the truncated Laplacian mechanism significantly reduces the noise amplitude and noise power in all privacy regimes. The improvement is not very surprising, as the truncated Laplacian mechanism universally improves the probability density decay rate (for both small and big noises) and thus leads to smaller noise amplitude and noise power in expectation.

\begin{figure}[H]
    \centering
    \includegraphics[width=0.45\textwidth]{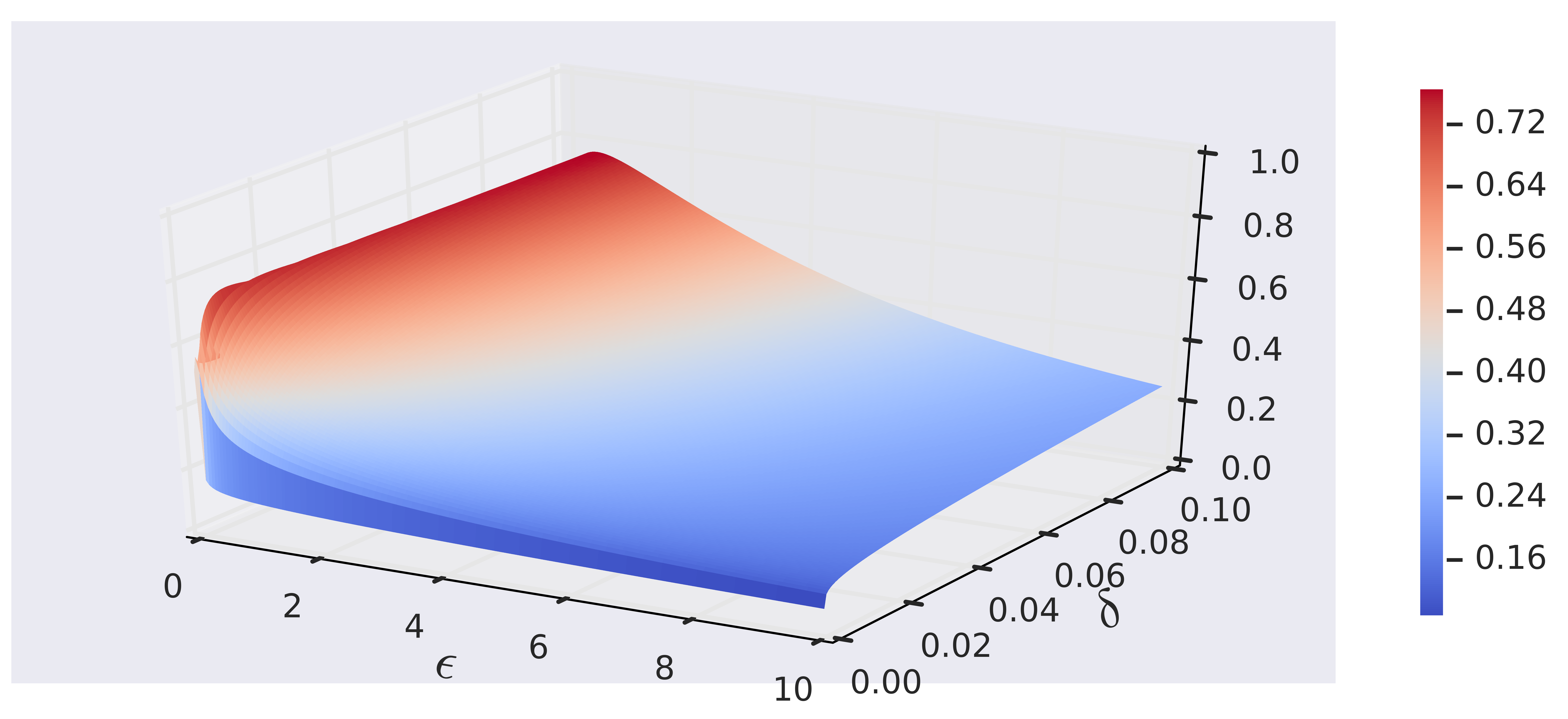}
    \caption{Ratio of the Noise Power of the Truncated Laplacian Mechanism and the Optimal Gaussian Mechanism.}
    \label{fig:ratio_l2}
\end{figure}

\section{Conclusion and Discussion} \label{sec:conclusion}
In this work, we characterize the minimum noise amplitude and noise power for noise-adding mechanisms in $(\e, \delta)$-differential privacy for single real-valued query function. We derive new lower bounds using the duality of linear programming, and derive new upper bounds by proposing a new class of $(\e,\delta)$-differentially private mechanisms, the \emph{truncated Laplacian} mechanisms. We show that the multiplicative gap of the lower bounds and upper bounds goes to zero in various high privacy regimes, proving the tightness of the lower and upper bounds and thus establishing the optimality of the truncated Laplacian mechanism. In particular, our results close the previous constant multiplicative gap in \citet{GV_IT_Approximate}. Comprehensive numeric experiments show the improvement of the truncated Laplacian mechanism over the optimal Gaussian mechanism in \citet{icmlGaussian} in all privacy regimes.

An obvious question is how to further improve the truncated Laplacian mechanism to provide stronger privacy guarantees. 
To minimize the additive noise, an important property of the truncated Laplacian mechanism is that the range of the output noise is bounded between $[-A, A]$. Therefore, for two neighboring datasets, the randomized output ranges will have some non-overlapped set. While the truncated Laplacian mechanism can strictly preserve $(\e, \delta)$-differential privacy, with a small probability up to $\delta$ (corresponding to the probability that the output is in the non-overlapped set), an adversary can distinguish the two neighboring datasets. To address this concern, one can improve over the truncated Laplacian mechanism and impose an arbitrarily light tail distribution over $[A, +\infty)$ to ensure that the output space is the same for all possible datasets. 

\bibliography{reference}
\bibliographystyle{icml2019}

\end{document}